\documentclass[12pt,a4paper]{article}

\usepackage{amsmath,amsfonts, amssymb, graphics, graphicx}

\begin{document}
\bibliographystyle{abbrv}

\title{Interference alignment-based sum capacity bounds for random dense Gaussian
interference networks}
\author{
Oliver Johnson\thanks{Department of Mathematics, University of Bristol, University Walk, Bristol, BS8 1TW, UK.
Email: \texttt{O.Johnson@bristol.ac.uk}}
\and Matthew Aldridge\thanks{Department of Mathematics, University of Bristol, University Walk, Bristol, BS8 1TW, UK.
Email: \texttt{m.aldridge@bristol.ac.uk}}
\and
Robert Piechocki\thanks{Centre for Communications Research,
University of Bristol, Merchant Venturers Building,
Woodland Road, Bristol
BS8 1UB, UK. Email: \texttt{r.j.piechocki@bristol.ac.uk}}
}
\date{\today}
\maketitle

\renewcommand\labelenumi{(\roman{enumi})}
\renewcommand\theenumi{(\roman{enumi})}
\newtheorem{theorem}{Theorem}[section]
\newtheorem{lemma}[theorem]{Lemma}
\newtheorem{proposition}[theorem]{Proposition}
\newtheorem{corollary}[theorem]{Corollary}
\newtheorem{conjecture}[theorem]{Conjecture}
\newtheorem{definition}[theorem]{Definition}
\newtheorem{example}[theorem]{Example}
\newtheorem{condition}{Condition}
\newtheorem{main}{Theorem}
\newtheorem{remark}[theorem]{Remark}
\hfuzz25pt

\def \outlineby #1#2#3{\vbox{\hrule\hbox{\vrule\kern #1%
\vbox{\kern #2 #3\kern #2}\kern #1\vrule}\hrule}}%
\def \endbox {\outlineby{4pt}{4pt}{}}%
\newenvironment{proof}
{\noindent{\bf Proof\ }}{{\hfill \endbox
}\par\vskip2\parsep}
\newenvironment{pfof}[2]{\removelastskip\vspace{6pt}\noindent
 {\it Proof  #1.}~\rm#2}{\par\vspace{6pt}}

\newcommand{\vc}[1]{{\mathbf{ #1}}}
\newcommand{\var}{{\rm{Var\;}}}
\newcommand{\cov}{{\rm{Cov\;}}}
\newcommand{\tends}{\rightarrow \infty}
\newcommand{\C}{{\cal C}}
\newcommand{\ep}{{\mathbb {E}}}
\newcommand{\pr}{{\mathbb {P}}}
\newcommand{\re}{{\mathbb {R}}}
\newcommand{\I}{\mathbb {I}}
\newcommand{\Z}{{\mathbb {Z}}}
\newcommand{\Nat}{{\mathbb {N}}}

\newcommand{\blah}[1]{}

\newcommand{\AAA}{{\mathcal{A}}}
\newcommand{\BBB}{{\mathcal{B}}}
\newcommand{\D}{{\mathcal{D}}}
\newcommand{\SSC}{{\mathcal{S}}}
\newcommand{\SSspine}{{\mathcal{S}}_{\rm spine}}
\newcommand{\SSedge}{{\mathcal{S}}_{\rm edge}}
\newcommand{\SSbody}{{\mathcal{S}}_{\rm body}}
\newcommand{\TTC}{\overline{\mathcal{S}}}
\newcommand{\RR}{{\mathcal{R}}}
\newcommand{\TT}{{\mathcal{T}}}
\newcommand{\SNR}{{\rm{SNR}}}
\newcommand{\INR}{{\rm{INR}}}
\newcommand{\R}[1]{R[#1]}
\newcommand{\RC}[1]{R_C[#1]}
\newcommand{\RS}[1]{R^*[#1]}
\newcommand{\bino}{{\rm{Bin}}}
\newcommand{\csum}{C_{\Sigma}}
\newcommand{\euc}{d}
\newcommand{\half}{\frac{1}{2}}
\newcommand{\ol}[1]{\overline{#1}}
\newcommand{\csep}{C_{\rm{sep}}}
\newcommand{\dsep}{D_{\rm{sep}}}
\newcommand{\cdec}{C_{\rm{dec}}}
\newcommand{\Fsnr}{F_{S}}
\newcommand{\Fdis}{F_{\euc}}
\newcommand{\fdis}{f_{\euc}}
\newcommand{\idc}[2]{\vc{#1},\vc{#2}}
\newcommand{\B}[2]{B_{\idc{#1}{#2}}}
\newcommand{\pB}[2]{p_{\idc{#1}{#2}}}
\newcommand{\prj}[1]{\Pi_{\vc{#1}}}

\begin{abstract}
\noindent
 We consider a dense 
$K$ user Gaussian interference network formed by paired transmitters and
receivers placed independently at random in a fixed spatial 
region. Under natural conditions on the 
node position distributions and signal attenuation, we
prove convergence in probability of the average per-user capacity $\csum/K$ to
$\half \ep \log(1 + 2 \SNR)$.
The achievability result follows directly from results based on an interference
alignment scheme presented in recent work of Nazer et al. Our main contribution comes through an upper bound, motivated by ideas of `bottleneck 
capacity' developed in recent work of Jafar. By controlling the physical
location of transmitter--receiver pairs, we can match a large proportion of 
these pairs to form so-called $\epsilon$-bottleneck links, with consequent 
control of the sum capacity.
\end{abstract}
\section{Introduction and main result}
\subsection{Interference networks and bottleneck states}
Recent work of Jafar \cite{jafar} made significant progress towards
what is referred to as `the holy grail of network information theory', namely
the calculation of the capacity of arbitrary Gaussian interference networks.
Jafar proves convergence in probability of the averaged sum capacity
of certain dense Gaussian interference networks. Although results contained in the
paper \cite{jafar} made significant progress with this problem, 
the results were described under
the constraint that each direct link had the same fading coefficient $\sqrt{\SNR}$ --
a constraint that we relax in this paper. 

In Lemma 1 of \cite{jafar}, Jafar showed that a two-user Gaussian interference channel with one
of the cross-link strengths $\INR = \SNR$ has sum capacity exactly equal to
$\log(1+2\SNR)$. 
Jafar described such a configuration as an example of 
a `bottleneck state', in that altering the other cross-link
strength does not affect the capacity. Jafar went on to
define the
concept of an $\epsilon$-bottleneck link -- that is, a cross-link in a two-user
channel with capacity within
$\epsilon$ of $\log(1+ 2 \SNR)$.
He considered a model of  large networks,
where each value of $\SNR$ is fixed, and each $\INR$ is sampled independently from
a fixed distribution. Jafar argues that with probabilty $\delta$, each $\INR$ lies
in the range such that the corresponding two-user channel becomes an 
$\epsilon$-bottleneck. 
 He uses an argument based on Chebyshev's inequality to deduce 
that  
\begin{equation} \label{eq:jafar}
\pr \left( \left| \frac{\csum}{K} - \half \log(1+2 \SNR) \right| > \epsilon
\right) = O(K^{-2}).\end{equation} 
It is perhaps surprising that
the existence of a positive proportion
of $\epsilon$-bottleneck links implies accurate probabilistic bounds
on the sum capacity $\csum$ of the whole network. However, we might regard it as
analogous to the so-called `birthday paradox' -- although each cross-link has a probability
$\delta$ of being in an $\epsilon$-bottleneck, there are $K(K-1)$ cross-links that 
can have this property, so as $K$ tends to infinity, the number of cross-links with
this property becomes much larger than $K$.

In this paper, we show that results such as Equation (\ref{eq:jafar}) in fact hold more
generally, in cases where transmitter node positions $T_1, \ldots, T_K$ and
corresponding receiver node positions $R_1, \ldots, R_K$ are chosen independently 
at random in a region of space $\D$. 
While exact expressions for the capacity
of networks with arbitrarily placed nodes remain elusive, in Theorem
\ref{thm:main} we prove convergence
in probability for the average per-user capacity $\csum/K$ 
of such dense network configurations. This may perhaps
be analogous to the fact that  Shannon \cite{shannon} proved that the average
code performed well, while it remains a significantly harder problem to establish
good performance for a particular family of codes. In other words,  the randomness
we add to the model helps us, rather than making things harder.
The intuition is that in a dense network
of points, a large proportion of nodes can be put together pairwise to
form $\epsilon$-bottleneck links (see Figure \ref{fig:bottle}). 

\begin{figure}[ht]
\begin{center}
\begin{picture}(262,262)(19,19)
\put(20,20){\line(0,1){260}}
\put(280,20){\line(0,1){260}}
\put(20,20){\line(1,0){260}}
\put(20,280){\line(1,0){260}}
\put(30,30){\circle{3}}
\put(22,35){$\vc{T_5}$}
\put(30,30){\line(4,3){80}}
\put(110,90){\circle{3}}
\put(110,92){$\vc{R_5}$}
\put(110,192){\circle{3}}
\put(110,194){$\vc{T_1}$}
\put(110,90){\line(0,1){102}}
\put(170,112){\circle{3}}
\put(170,114){$\vc{R_1}$}
\put(110,194){\line(3,-4){60}}
\put(50,170){\circle{3}}
\put(50,172){$\vc{T_4}$}
\put(50,170){\line(1,0){200}}
\put(250,170){\circle{3}}
\put(250,172){$\vc{R_4}$}
\put(110,30){\circle{3}}
\put(93,32){$\vc{T_2}$}
\put(250,170){\line(-1,-1){140}}
\put(270,150){\circle{3}}
\put(260,155){$\vc{R_2}$}
\put(110,30){\line(4,3){160}}
\put(250,60){\circle{3}}
\put(250,62){$\vc{R_3}$}
\put(83,232){\circle{3}}
\put(83,234){$\vc{T_3}$}
\put(25,90){\circle{3}}
\put(25,92){$\vc{T_6}$}
\put(220,193){\circle{3}}
\put(220,195){$\vc{R_6}$}
\end{picture}
\caption{A dense network with $K=6$ transmitter--receiver pairs placed on the square $[0,1]^2$, with $\epsilon$-bottleneck links emphasised. The distances
from $T_5$ to $R_5$, $R_5$ to $T_1$ and $T_1$ to $R_1$ are all approximately equal,
similarly for $T_4$ to $R_4$, $R_4$ to $T_2$ and $T_2$ to $R_2$. Transmitter-receiver 
pairs $(T_3,R_3)$ and $(T_6,R_6)$ are not matched into
$\epsilon$-bottleneck links.\label{fig:bottle}}
\end{center}
\end{figure}
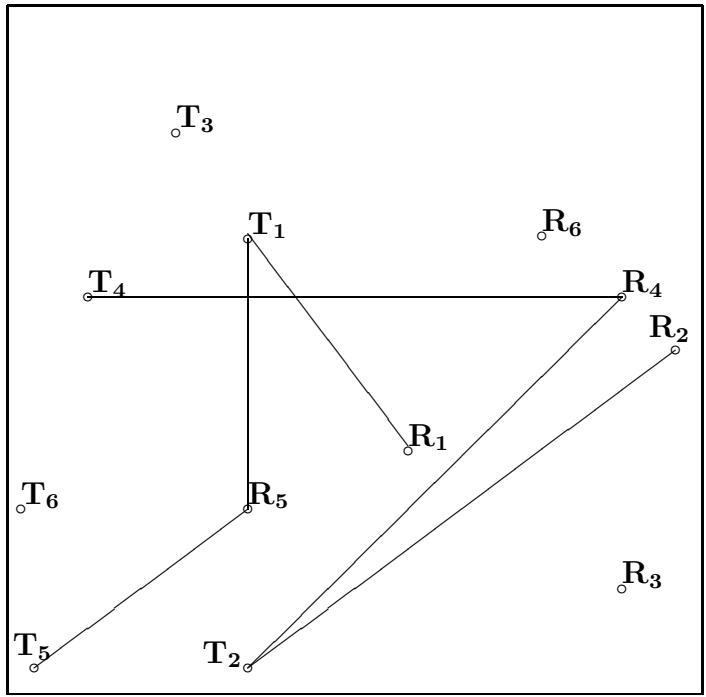
While it would be possible to adjust
individual transmitters' powers to force each $\SNR$ to become equal (as assumed by
\cite{jafar}), individual
power constraints make this undesirable. Further, Jafar assumed that the $\INR$
are independent and identicaly distributed, a property that would be lost if user
powers were scaled in this way.

\subsection{Interference Alignment }

The concept of interference alignment first appeared in 
\cite{cadambe}, and represents a departure from the recent
paradigm of random (or rather pseudo-random) code construction. 
Pseudo-random
codes (e.g. Turbo and LDPC codes) have revolutionised point-to-point communications.
However, since multi-terminal networks are interference-limited rather than
noise-limited, unstructured (pseudo-random) codes are not suitable.

Interference alignment advocates a collaborative solution. Each receiver
divides its signalling space (space/time/frequency/scale resources)
into two  parts; one for the signal from the intended
transmitter, and the second acts as a waste bin. The encoding is structured
is such a way that the transmitted signal from each of the $K$ transmitters
is seen in the clear space for the intended receiver, and at the same
time, it falls into the waste bin for all other receivers. In such a
scenario, the network is no longer interference limited. In the first
such scheme Cadambe and Jafar \cite{cadambe}
showed that spatio/temporal beamforming
in the $\SNR\rightarrow\infty$ regime allows to achieve for each user
pair a {}``capacity'' equal to half that 
of the single user channel. A different
idea recently appeared in \cite{nazer},
which does not require $\SNR\rightarrow\infty$. Consider
the following channel pair:
\begin{equation}
\mathbf{H_{a}}=\left[\begin{array}{ccc}
1 & -1 & 1\\
1 & 1 & -1\\
-1 & 1 & 1\end{array}\right],\;\;\; 
\mathbf{H_{b}}=\left[\begin{array}{ccc}
1 & 1 & -1\\
-1 & 1 & 1\\
1 & -1 & 1\end{array}\right]\label{eq:martix_pair}\end{equation}
The channels are constructed in such a way that the simple sum of the
received signals leads to interference alignment since $\mathbf{H_{a}}+\mathbf{H_{b}=\mathrm{2}I}$.
In more general channels, such as those with Rayleigh fading, 
one needs to code
over sufficiently long time intervals to observe and match complementary
matrix pairs. 
\subsection{Node positioning model}
We believe that our techniques should work in a variety of models for the node positions. We outline
one very natural
scenario here.
\begin{definition} \label{def:nodeplace}
Consider a fixed spatial region $\D \in \re^D$, with two probability distributions $\pr_T$ 
and $\pr_R$ supported on $\D$. Given an integer $K$, we sample the 
$K$ transmitter node positions
$T_1, \ldots, T_K$ independently from the distribution $\pr_T$. Similarly, we sample
the $K$ receiver node $R_1, \ldots, R_K$ positions independently from distribution $\pr_R$. We
refer to such a model of node placement as an `IID network'.
\end{definition}
Equivalently, we could state that transmitter and receiver positions are distributed according to two independent
(non-homogeneous) Poisson processes, conditioned such that there are $K$  points 
of each type in $\D$.
 We pair the transmitter and receiver nodes up so that $T_i$ wishes to communicate 
with $R_i$ for each $i$.
We make the following definition:
\begin{definition} \label{def:spatsep}
We say that transmitter and receiver distributions $\pr_T$ and $\pr_R$
are `spatially separated' if
there exist  constants $\csep < \infty$ and $\dsep < \infty$
such that for $T \sim \pr_T$ and $R \sim \pr_R$ the
Euclidean distance $\euc(T,R)$ satisfies
\begin{equation} \label{eq:spatsep}
  \pr( \euc( T, R) \leq s) \leq \csep s^{\dsep} \mbox{ \;\;\; for all $s$}.\end{equation}
\end{definition}
We argue in Lemma \ref{lem:spatsep} below that a wide range of node distributions $\pr_T$ and
$\pr_R$ have this spatial separation property, which allows us to control the tails of
the distribution of $\SNR$, and hence the maximum value of $\SNR$ in Lemma \ref{lem:tailbd}
\subsection{Transmission models}
For simplicity, we first describe our results in the context of so-called `line of sight'
communication models, without multipath interference. That is,
we consider a model where signal strengths decay 
deterministically with Euclidean distance $\euc$ according to
some monotonically decreasing continuous
function $f(\euc)$. 
We make the following definition, which complements the definition of
spatial separation given in Definition \ref{def:spatsep}.
\begin{definition} \label{def:decay}
We say that the signal is `decaying at rate $\alpha$' if there exists $\cdec < \infty$
such that for all $\euc$
\begin{equation} \label{eq:decay}
 f(\euc)  \leq \cdec \euc^{-\alpha}.\end{equation}
\end{definition}
Standard physical considerations imply that all signals must be decaying 
at some rate $\alpha \geq D$, where $D$ is the dimension of the underlying space.
Tse and Viswanath \cite[Section 2.1]{tse} discuss a variety of models under which this 
condition holds for different exponents $\alpha$.

We define the full action of the Gaussian interference network:
\begin{definition}
\label{def:transprot}
Fix transmitter node positions $\{ T_1, \ldots, T_K \} 
\in \D$ and receiver node positions
$\{R_1, \ldots, R_K \} \in \D$, and consider Euclidean distance $\euc$ and attenuation function $f$.
For each $i$ and $j$, define 
$\INR_{ij} = f(\euc(T_i,R_j))$.
For emphasis, for each $i$ we write $\SNR_i$ for $\INR_{ii}$.

We consider the $K$ user Gaussian interference network defined so that
transmitter $i$ sends a message encoded as a sequence of complex
numbers $\vc{X}_i = (X_i[1], \ldots, X_i[N])$ to receiver $i$, under a power constraint
$\frac{1}{N} \sum_{n=1}^N |X_i[n]|^2
 \leq 1$ for each $i$. The $n$th symbol received at receiver $j$ is given as
\begin{equation} \label{eq:transmod} Y_j[n] = \sum_{i=1}^K \exp(i \phi_{ij}[n]) 
\sqrt{\INR_{ij}} X_i[n]  + Z_j[n],
\end{equation}
where $Z_j[n]$ are independent standard complex Gaussians, and $\phi_{ij}[n]$ are
independent $U[0,2\pi]$ random variables independent of all other terms.
The $\INR_{ji}$ remain fixed, since the node positions themselves are fixed,
but the phases are fast fading.
\end{definition}
We write $S_{ij}$ for the random variables $\half \log(1 + 2 \INR_{ij})$,
which  are functions of the distance
between $T_i$ and $R_j$.
 In particular, since the nodes are positioned
independently in Definition \ref{def:nodeplace}, under this model
the random variables $S_{ii} = \half \log(1 + 2 \SNR_i)$ are independent
and identically distributed. 

In Section \ref{sec:randfad}, we explain how our techniques can be extended to apply to more general models, in the presence of random fading terms.
\subsection{Main result: convergence in probability of $\csum/K$}
We now state the main theorem of this paper, which proves convergence in
probability of the averaged capacity, under the model of node placement
described in Definition \ref{def:nodeplace}  and 
the model for signal attenuation described in Definition
\ref{def:transprot}. For the sake of clarity, we restrict our attention to the
case where $\pr_R$ and $\pr_T$ are uniform, though we
discuss later to what extent this assumption is necessary. 

We restrict to bounded
regions $\D$ with a smooth boundary -- the sense of this smoothness will be made
precise in the proof of Theorem \ref{thm:main}. Theorem \ref{thm:main} will certainly
hold for squares and balls $\D = [0,1]^D$ and $\D = \{ \vc{x}: \euc(\vc{x}, \vc{0})
\leq 1 \}$, and indeed for any convex and bounded polytopes (with finite surface area). 
Essentially we require that the boundary of $\D \times \D$
has Hausdorff dimension $\leq 2D-1$, which is very natural.
\begin{theorem} \label{thm:main}
Consider a Gaussian interference network formed by $K$ pairs of 
nodes placed in an IID network, with the signal decaying at some rate $\alpha \geq D$.
If distributions $\pr_T$ and $\pr_R$ are both uniform on a bounded region $\D$ with
smooth boundary, then
the average per-user capacity $\csum/K$ converges in probability to 
$\half \ep \log(1 + 2 \SNR)$, that is 
$$ \lim_{K \tends} \pr \left( \left| \frac{\csum}{K} - 
\half \ep \log(1 + 2 \SNR)
 \right| > \epsilon \right)
= 0 \mbox{\;\;\;\; for all $\epsilon > 0$}.$$
\end{theorem}
\begin{proof}
We  break the probability into two terms which we deal
with separately. That is, writing
$E = \ep S_{ii} = \half \ep \log(1 + 2 \SNR)$,
\begin{equation} \label{eq:unionbd}
 \pr \left( \left| \frac{\csum}{K} - E \right| > \epsilon \right)
= \pr \left( \frac{\csum}{K} - E  < -\epsilon \right) +
\pr \left( \frac{\csum}{K} - E  > \epsilon \right).\end{equation}

Bounding the first term of (\ref{eq:unionbd}) corresponds to the achievability part of the proof.
Bounding the second term of (\ref{eq:unionbd}) corresponds to the converse part, and
represents our major contribution.

The first term of (\ref{eq:unionbd}) 
can be bounded relatively simply, using an achievability
argument based on an interference alignment scheme presented 
by Nazer, Gastpar, Jafar and Vishwanath \cite{nazer}. 
Theorem 3 of \cite{nazer} implies that the rates $\R{i} =
1/2 \log(1 + 2 \SNR_i) = S_{ii}$ are achievable. This implies that
$\csum \geq \sum_{i=1}^K S_{ii}$. This allows us to bound the first term in
Equation (\ref{eq:unionbd}) as 
\begin{eqnarray}
\pr \left( \frac{\csum}{K} - E  < -\epsilon \right)
& \leq &  \pr \left( \frac{\sum_{i=1}^K (S_{ii}-E)}{K}   < -\epsilon \right).
\label{eq:lln}
\end{eqnarray}
Note that Lemma \ref{lem:spatsep} below implies that the spatial
separation condition holds in the setting of Theorem \ref{thm:main}.
Hence the conditions of Lemma \ref{lem:tailbd} hold, implying that $S_{ii}$ has finite variance
by Equation (\ref{eq:exptail}).
This means that Equation (\ref{eq:lln})
 can be bounded by $\var(S_{ii})/(K \epsilon^2)$, and tends to zero
at rate $O(1/K)$.

We consider the properties of the second term of Equation (\ref{eq:unionbd}) in
the remainder of the paper, completing the proof of the theorem at the 
end of Section \ref{sec:proof}. 
\end{proof}
Theorem \ref{thm:main} can be interpreted in the same way as Theorem 1 of \cite{jafar},
that `each user is able to achieve the same rate that he would achieve if he
had the channel to himself with no interferers, half the time'. 

The theorem is presented
under the assumptions of uniform $\pr_R$ and $\pr_T$ with deterministic fading
for the sake of simplicity of exposition -- we believe that the main
result will be robust to relaxation of these conditions.
 In Section \ref{sec:non-unif} we consider whether the $\pr_R$ and $\pr_T$ need
necessarily be uniform. In Section \ref{sec:randfad} we introduce a variant of the model
with a random fading term. Although the theorem is based on probabilistic arguments, 
 in Section \ref{sec:algorithm} we describe an associated algorithm which
bounds the sum capacity of arbitrary Gaussian interference networks.
\subsection{Relation to previous work}
As reviewed in more detail by Jafar \cite{jafar}, recently  
progress has been made in several directions towards understanding 
the capacity of Gaussian interference networks.

In problems concerning networks with a large number of 
nodes, work of Gupta and Kumar \cite{gupta} uses
techniques based on Voronoi tesselations to establish scaling laws
(see also Xue and Kumar \cite{xue} for a review of the information theoretical
techniques that can be applied to this problem). Under a similar model of
dense random network placements, though using the same points as both transmitters
 and receivers,
  {\"O}zg{\"u}r, L{\'e}v{\^e}que and Tse \cite{ozgur2, ozgur} use a hierarchical scheme, where nodes are successively
assembled into groups of increasing size, each group collectively acting as a MIMO
transmitter or receiver, and restricting to transmissions at a common
rate. Theorems 3.1 and 3.2 of \cite{ozgur} show that
for any $\epsilon > 0$ there
exists a constant $c_\epsilon$ and a fixed constant $c_1$ such that 
\begin{equation} \label{eq:ozgur}
c_\epsilon K^{1-\epsilon} \leq \csum \leq c_1 K \log K.
\end{equation}
These bounds are close to stating that $\csum$
grows like $K$,  but without 
the explicit constant
that Jafar \cite{jafar} and Theorem \ref{thm:main} of this paper achieve. 
In this paper, we produce a version of the upper bound of 
Equation (\ref{eq:ozgur}) without the logarithmic factor 
and being explicit about the constant $c_1$, although note that this result is proved under a model 
that differs from that of \cite{ozgur} in the fact that we have a total of $2K$ nodes rather than $K$.

An alternative approach to Gaussian interference networks is to consider the limit of the capacity
as the $\SNR$
tends to infinity, with a fixed number of users. Cadambe and Jafar \cite{cadambe} 
used interference alignment to deduce the limiting behaviour within
$o(\log(\SNR))$.  These techniques were extended by the same authors \cite{cadambe2}
to more general models  in the presence of feedback and other effects.

For small networks, the classical bounds due to Han and Kobayashi \cite{han2}
for the two-user Gaussian interference network
have recently been extended and refined. For example Etkin, Tse and Wang
\cite{etkin} have produced a characterization of capacity accurate to within one
bit. These results were extended by Bresler, Parekh and Tse \cite{bresler},
using insights based on a deterministic channel which approximates the Gaussian
channel with sufficient accuracy, to prove results for many-to-one and one-to-many
Gaussian interference channels.

We briefly mention an alternative proof of Theorem 5 of \cite{jafar}, which provides
a faster rate of convergence than that achieved in Equation (\ref{eq:jafar}). 
We first review some facts from graph theory, concerning
random bipartite graphs formed by
two sets of vertices of size $N$, with edges present independently with probability
$\delta$. Erd\"{o}s and R\'{e}nyi \cite{erdos} prove that
the probability of a complete matching 
failing to exist tends to $0$ for any $\delta = \delta(K) = (\log K + c_K)/K$,
where $c_K \rightarrow \infty$. 
We recall  the argument where $\delta$ is fixed, so that we
can be precise about the bounds, rather than just working asymptotically.
As in, for example, Walkup \cite{walkup2}, we say that a subset $\AAA_S \subset 
\AAA$ of
size $k$ and a subset $\BBB_S \subset \BBB$ of size $N-k+1$ form a blocking pair of 
size $k$ if no
edge of the graph connects $\AAA_S$ to $\BBB_S$. Equation (1) of \cite{walkup2} uses
K\"{o}nig's theorem to deduce
that 
\begin{eqnarray} 
\pr( \mbox{no matching $\AAA$ to $\BBB$} ) & \leq &
\sum_{k=1}^N \sum_{|\AAA_S| = k, |\BBB_S| = N-k+1}  \pr \left( (\AAA_S, \BBB_S) 
\mbox{ blocking pair} \right) \label{eq:walkup} \\
& = & 2 \sum_{k=1}^{(N+1)/2} \binom{N}{k} \binom{N}{k-1} (1-p)^{k(N-k+1)}. \nonumber
\end{eqnarray}
By splitting the sum into terms where $k \leq \sqrt{N}$ and $k \geq \sqrt{N}$, 
so that $(1-p)^{k(N-k+1)}$ is bounded by $\exp(-p (N+1)/2)$ and $\exp(-p N^{3/2}/2)$
respectively, a bound of
$$ 2 \sqrt{N} N^{2 \sqrt{N}} \exp(-p(N+1)/2) + 2^{2N} \exp(-p N^{3/2}/2)$$ 
can be obtained.
We deduce that the probability
of a complete matching failing to exist decays at an exponential rate in $N$.

To prove Theorem 5 of \cite{jafar},
we divide the receiver-transmitter links 
into two groups $\AAA$ and $\BBB$ of 
size $N = \lfloor K/2 \rfloor$, and consider complete matchings on the bipartite graph between
them. Each edge is present in the bipartite graph if the corresponding $\INR$ lies 
in a particular range (see Lemma \ref{lem:twouser} below for details), 
which occurs independently with probability
$\delta$.
For each pair that is successfully matched up, the corresponding two-user channel becomes an 
$\epsilon$-bottleneck, and
contributes $\leq \log (1+ 2 \SNR) + \epsilon$ to the sum
capacity. Hence, the high probability of a complete matching implies exponential decay of
$\pr(|\csum/K - \half \log(1+2 \SNR)| > \epsilon)$, improving the $O(K^{-2})$ rate 
in Equation (\ref{eq:jafar}).
\section{Technical lemmas} \label{sec:proofupper}
We continue to work towards 
our proof of Theorem \ref{thm:main}, by establishing some technical results. First in Section
\ref{sec:condition}, we identify a condition under which the sum capacity of
the two user interference channel can be bounded. Next, in Section \ref{sec:tails},
 we show that the spatial separation condition of
Definition \ref{def:spatsep} holds under a variety of conditions. Further we show that
spatial separation 
and the decaying condition Definition \ref{def:decay}
together imply that 
no $\SNR$ can be `too large'. 
\subsection{Bounds on two user channel} \label{sec:condition}
First, we identify a condition on the values of 
$\SNR$ and $\INR$ under which the capacity of the two user 
interference channel can be bounded.
The proof of the following result is based on Lemma 1
of \cite{jafar}, which gave the key definition of a `bottleneck state', deducing
that in the case $\SNR_i = \SNR_j = \INR_{ji} = \SNR$, the sum capacity equals
$\log(1+2\SNR)$. We reproduce the argument used there, to deduce a stability
result that allows us to deduce when an $\epsilon$-bottleneck state occurs.
\begin{lemma} \label{lem:twouser}
For any $i$,$j$, consider the two user inferference channel defined
by
\begin{eqnarray*}
Y_i &= & \exp(i \phi_{ii}) \sqrt{\SNR_i} X_i + \exp(i \phi_{ji}) \sqrt{\INR_{ji}} X_j + Z_i, \\
Y_j &= & \exp(i \phi_{ij}) \sqrt{\INR_{ij}} X_i + \exp(i \phi_{jj})  \sqrt{\SNR_j} X_j + Z_j,
\end{eqnarray*}
where $Z_i$ and $Z_j$ are IID standard complex Gaussians,  and $\phi_{ij}$ are
independent $U[0,2\pi]$ random variables independent of all other terms.

If $\INR_{ji} \geq \SNR_{j}$ then any reliable transmission rates
satisfy
\begin{equation} \label{summedconstraints}
\R{i} + \R{j} \leq \log(1 + \INR_{ji} + \SNR_i).\end{equation}
\end{lemma}
\begin{proof}
We adapt the argument of Lemma 1 of \cite{jafar}. That is, consider
reliable transmission rates $\R{i}$, $\R{j}$. Since the transmissions are
reliable, then receiver $i$ can determine $X_i$ with an arbitrarily low
probability of error.

Again, reliable transmission
rates would remain reliable if receiver $j$ was presented with $X_i$ by a 
genie. In that case, it is easier for receiver $i$ to determine $X_j$ than it is 
for receiver $j$ to do so (since the weighting $\INR_{ji}$ is larger than $\SNR_j$).
However, we know that receiver $j$ can recover $X_j$, since the rate $\R{j}$ is 
reliable, so we deduce that
receiver $i$ must be able to do the same.

Since receiver $i$ can determine $X_i$ and $X_j$ reliably, then these messages
must have been transmitted at a sum rate lower than the sum capacity of
a two-user multiple access channel, see for example \cite[Equation (6.6)]{tse},
 which is $\log(1 + \SNR_i + \INR_{ji})$.
\end{proof}
\subsection{Decay of tails} \label{sec:tails}
Recall that the node positioning model given in
Definition \ref{def:nodeplace} involves independent positions of nodes
sampled from identical distributions $\pr_T$ and $\pr_R$. We give examples
of conditions under which the spatial separation property of Definition
\ref{def:spatsep} holds.
\begin{lemma} \label{lem:spatsep} \mbox{ }
\begin{enumerate}
\item \label{itm:ex1}
If either $\pr_T$ or $\pr_R$ has a density with respect to Lebesgue measure
which is bounded above on $\D$ then $\pr_T$ and $\pr_R$ are spatially separated.
\item \label{itm:ex2}
If $\pr_T$ and $\pr_R$ are supported on sets $\TT$ and $\RR$ that 
are physically separated, in 
that 
$$ \euc_* = \inf \{ \euc( t, r): t \in \TT, r \in \RR \} > 0,$$
then $\pr_T$ and $\pr_R$ are spatially separated.
\end{enumerate}
\end{lemma}
\begin{proof}
\ref{itm:ex1}
If $\pr_T$ has a density  with respect to Lebesgue measure which is bounded above
by $C$,
then for any ball $B_s(y)$ of radius $s$ centred on $y$, the probability $\pr_T( B_s(y))
\leq C V_D s^D$, where $V_D$ is the volume of a Euclidean ball of unit radius
in $\re^D$.
Hence
$$ \pr( \euc(T,R) \leq s) = \int  \pr_T( B_s(y)) d\pr_R(y) \leq C V_D s^D 
\int d\pr_R(y) = C V_D s^D,$$ and the result follows,
taking $\csep = C V_D$ and $\dsep = D$. The corresponding result for
$\pr_R$ follows on exchanging $\pr_R$ and $\pr_T$ in the displayed equation above.

\ref{itm:ex2} If $s < \euc_*$, then $\pr(\euc(T,R) \leq s) = 0$.  If $s \geq \euc_*$, then we have
$\pr( \euc(T,R) \leq s) \leq 1 \leq s^D/\euc_*^D$, so that we can take $\csep = 1/\euc_*^D$
and $\dsep = D$.
\end{proof}
Next we show that combining the spatial separation condition of
Definition \ref{def:spatsep} and the decaying condition Definition \ref{def:decay}
gives us good control of the maximum of $S_{ii}$. The argument is similar
to that given in Theorem 3.1 of \cite{ozgur}.
\begin{lemma} \label{lem:tailbd}
Consider an  IID network with spatially separated $\pr_T$ and 
$\pr_R$, with signals decaying at rate $\alpha$.
The probability that the maximum of the $K$ 
random variables $S_{ii}$ is large tends to zero:
\begin{equation} \label{eq:extreme}
 \lim_{K \tends} \pr \left( \max_{1 \leq i \leq K} S_{ii} \geq \frac{\alpha \log K}{\dsep}
\right) = 0, \end{equation}
where $\dsep$ is the separation exponent from 
Definition \ref{def:spatsep}.
\end{lemma}
\begin{proof}
 We combine Definitions \ref{def:spatsep} and \ref{def:decay}.
Since all $S_{ij}$ have the same marginal distribution, it is enough to deduce that, for
any $u \geq 1$,
\begin{eqnarray}
\pr( S_{ii} \geq u) & = & \pr( 1/2 \log(1 + 2 \SNR_i) \geq u) \nonumber \\
& = & \pr( \SNR_i \geq (\exp(2u) - 1)/2) \nonumber \\
& \leq & \pr( \SNR_i \geq \exp(2u)/3) \nonumber \\
& = & \pr( f(\euc(T_i, R_i)) \geq \exp(2u)/3) \nonumber \\
& \leq & \pr( \cdec/(\euc(T_i, R_i))^{\alpha} \geq \exp(2u)/3) \nonumber \\
& = & \pr( \euc(T_i, R_i) \leq (\cdec/3)^{1/\alpha} \exp(-2u/\alpha)) \nonumber \\
& \leq & \csep (\cdec/3)^{\dsep/\alpha} \exp(-2u \dsep/\alpha)), \label{eq:exptail}
\end{eqnarray}
where $\csep$ is the separation 
constant from Equation (\ref{eq:spatsep}) and $\cdec$ is the decay 
constant from Equation (\ref{eq:decay}).
The result follows from Equation (\ref{eq:exptail}) using the union bound:
\begin{eqnarray*}
\pr \left( \max_{1 \leq i \leq K} S_{ii} \geq \frac{\alpha \log K}{\dsep} \right) 
& \leq & K \pr \left( S_{ii} \geq \frac{\alpha \log K}{\dsep} \right) 
\leq  K \frac{\csep (\cdec/3)^{\dsep/\alpha}}{K^2},
\end{eqnarray*}
which tends to zero as $K \tends$. 
\end{proof}
\section{Proof of Theorem \ref{thm:main}} \label{sec:proof}
In this section we complete the proof of the upper bound in Theorem
\ref{thm:main}. 
We calculate bounds on sum capacity using a strategy suggested by the work
of Jafar \cite{jafar}, who proves that the presence of a large number of disjoint 
two-user channels, each close to being a `bottleneck state', allows good control 
of the channel capacity. 

First in Section \ref{sec:spatial} we partition space
into regions $\B{u}{v}$.   Lemma \ref{lem:boxing} gives
an upper bound on the sum capacity of a two-user channel made up of points
in neighbouring regions.
Further, Lemma \ref{lem:full}
tells us that we can control the number of 
links in each region.
In Section \ref{sec:matching} we complete the argument by
matching elements of box $\B{u}{v}$ with elements of $\B{u-e}{v+e}$.
This allows us to control the overall sum capacity of the $K$ users, and 
to complete the proof of Theorem \ref{thm:main}. 
\subsection{Spatial partitioning model} \label{sec:spatial}
We consider each transmitter--receiver pair as a point in the joint
domain $\D \times \D$. 
Each transmitter--receiver pair $(T_i,R_i) \in
\D \times \D$ can
be placed in well-defined disjoint regions $\B{u}{v}$, allowing us to control the
performance of the corresponding link.
We write $x^{(l)}$ for the $l$th coordinate of the vector $\vc{x} =
(x^{(1)}, \ldots, x^{(D)})$.
\begin{definition}
Given $M$, we partition the space $\re^{2D}$ and hence the 
joint transmitter--receiver domain $\D \times \D$ 
by a regular grid of spacings $1/M$.
For each $\vc{u} \in \Z^{D}$ and  $\vc{v} \in \Z^{D}$, 
we define boxes labelled by their `bottom-left' corner as
\begin{eqnarray} \label{eq:boxdef}
\B{u}{v} & = & \left\{ (\idc{x}{y}):   
\frac{u^{(l)}}{M} \leq x^{(l)} < \frac{(u^{(l)} +1)}{M},
\frac{v^{(l)}}{M} \leq y^{(l)} < \frac{(v^{(l)} +1)}{M} \mbox{ \; for all $l$} \right\}.
\end{eqnarray} 
We write $\SSC = \{ (\idc{u}{v}):  \B{u}{v} \bigcap \left( \D \times \D  \right)
\neq \emptyset \}$ for the set
of possible labels $(\idc{u}{v})$, and
split $\D \times \D$ into orthants, indexed by vectors $\vc{E}
(\idc{u}{v})
\in \{-1, 0, 1 \}^D$, with co-ordinate 
$$ E^{(l)}(\idc{u}{v}) = \left\{ \begin{array}{ll}
1 & \mbox{ \;\;\; if $v^{(l)} - u^{(l)} > 0$,} \\
0 & \mbox{ \;\;\; if $v^{(l)} - u^{(l)} = 0$,} \\
-1 & \mbox{ \;\;\; if $v^{(l)} - u^{(l)} < 0$.} \\
\end{array} \right.$$

We introduce two subsets of $\SSC$ that we will not 
attempt to match, according to the rule that $\B{u}{v}$ is matched with 
$\B{u-e}{v+e}$, where 
$\vc{e} = E(\idc{u}{v})$.
\begin{enumerate}
\item First, the `spine' $\SSspine =
\{ (\idc{u}{v}): E^{(l)}(\idc{u}{v}) = 0 \mbox{ for some
$l$} \}.$
\item Second, the `edge' $\SSedge = \{ (\idc{u}{v}): \left(\B{u}{v} \bigcup \B{u-e}{v+e}
\right) \not\subseteq \D \times \D
 \}$. That is, the regions which overlap the boundary of $\D \times \D$, or which are matched with a
 region that overlaps the boundary.
\end{enumerate}
\end{definition}
The control of position obtained by matching links in two
neighbouring regions converts into control of the values of $\SNR$ and 
$\INR$, allowing the sum capacity of the two pairs to be bounded using Lemma
\ref{lem:twouser}.
Figure \ref{fig:match} gives a schematic diagram of a pair of boxes which we attempt to
match using the construction in Lemma \ref{lem:boxing}, plotting both transmitter and
receiver positions singly on $\D$ rather than jointly on $\D \times \D$.

\begin{center}
\begin{figure}[!htbp]
\centering
\begin{picture}(215,215)(18,18)
\put(20,20){\line(0,1){210}}
\put(230,20){\line(0,1){210}}
\put(20,20){\line(1,0){210}}
\put(20,230){\line(1,0){210}}
\multiput(200,50)(-30,30){2}{\line(0,1){30}}
\multiput(230,50)(-30,30){2}{\line(0,1){30}}
\multiput(200,50)(-30,30){2}{\line(1,0){30}}
\multiput(200,80)(-30,30){2}{\line(1,0){30}}
\multiput(80,140)(-30,30){2}{\line(0,1){30}}
\multiput(110,140)(-30,30){2}{\line(0,1){30}}
\multiput(80,140)(-30,30){2}{\line(1,0){30}}
\multiput(80,170)(-30,30){2}{\line(1,0){30}}
\put(220,70){\circle{3}}
\put(224,70){$\vc{R_j}$}
\put(195,100){\circle{3}}
\put(199,100){$\vc{R_i}$}
\put(91,145){\circle{3}}
\put(95,145){$\vc{T_i}$}
\put(60,185){\circle{3}}
\put(64,185){$\vc{T_j}$}
\put(178,28){\vector(1,1){20}}
\put(170,23){$\vc{u} - \vc{e}$}
\put(148,58){\vector(1,1){20}}
\put(138,53){$\vc{u}$}
\put(58,118){\vector(1,1){20}}
\put(48,113){$\vc{v}$}
\put(28,148){\vector(1,1){20}}
\put(10,143){$\vc{v} + \vc{e}$}
\end{picture}
\caption{Schematic plot of matched boxes, where boxes are labelled according to their
bottom-left corner.
We consider the case where $\D = [0,1]^2$, and plot the receiver and transmitter
positions on the same square, with $\vc{e} = E(\idc{u}{v}) = (-1,1)$. The key
property is that we can observe that $\euc(T_j,R_j) \geq \euc(T_j, R_i) \geq \euc(T_i, R_i)$.
We prove this rigorously in order to prove Lemma \ref{lem:boxing}. \label{fig:match}}
\end{figure}
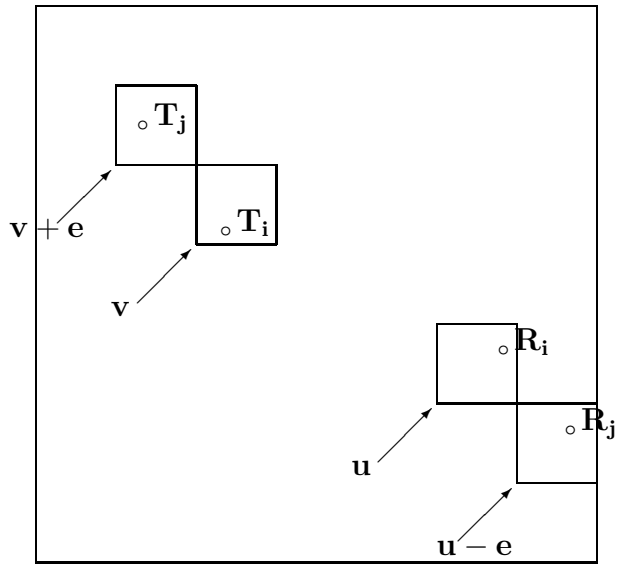
\end{center}
\begin{lemma} \label{lem:boxing}
Suppose the receiver-transmitter pair $(R_i,T_i)$ appears in region $\B{u}{v}$ and the 
receiver-transmitter pair $(R_j,T_j)$ appears in region $\B{u-e}{v+e}$,
where $(\idc{u}{v}) \in \SSC \setminus \SSspine$ and $\vc{e} = E( \idc{u}{v})$,
 then any reliable rates for those two links satisfy
$$ \R{i} + \R{j} \leq \log \left( 1 + 2f(\euc(\idc{u}{v})) \right),$$
where $\euc(\idc{u}{v})$ is the minimum transmitter--receiver 
distance between $(R_i,T_i) \in \B{u}{v}$.
\end{lemma}
\begin{proof} The rates are only improved by being presented with the 
messages of all the other users, reducing the situation to that of 
Lemma \ref{lem:twouser}. For each $l$ such that
co-ordinate $e^{(l)} = 1$, by construction
(a) $T_i^{(l)} - R_i^{(l)} > 0$, (b) $R_j^{(l)} < u^{(l)} \leq R_i^{(l)}$,
(c) $T_j^{(l)} \geq v^{(l)} > T_i^{(l)}$.
Hence, by (b),
$$(T_j^{(l)} - R_i^{(l)}) = (T_j^{(l)} - R_j^{(l)}) 
+ (R_j^{(l)} - R_i^{(l)}) < (T_j^{(l)} - R_j^{(l)}),$$
and by (c),
$$(T_j^{(l)} - R_i^{(l)}) = (T_i^{(l)} - R_i^{(l)}) 
+ (T_j^{(l)} - T_i^{(l)}) > (T_i^{(l)} - R_i^{(l)}).$$
Overall then, in the case $e^{(l)} = 1$,
$$ 0 < (T_i^{(l)} - R_i^{(l)}) < (T_j^{(l)} - R_i^{(l)}) < (T_j^{(l)} - R_j^{(l)}).$$

A similar argument applies for
each $l$ with $e^{(l)} = -1$, with the order of the signs reversed.

Overall, we deduce that $\euc(T_j,R_j) \geq \euc(T_j, R_i) \geq \euc(T_i, R_i) \geq
d(\idc{u}{v})$, or that
$f(d(\idc{u}{v})) \geq \SNR_i \geq \INR_{ji} \geq \SNR_j$, so that
Lemma \ref{lem:twouser} applies, allowing us to bound
$$ \R{i} + \R{j} \leq \log(1 + \INR_{ji} + \SNR_i) \leq \log
\left(1 + 2f(d(\idc{u}{v}) ) \right),$$
as required.
\end{proof}

 Each of the $K$ transmitter--receiver links are placed
in the regions $\B{u}{v}$ where $(\idc{u}{v}) \in \SSC$.
Our node positioning
rule implies that each link is placed in region
$\B{u}{v}$ with some probability $\pB{u}{v}$ independently of the others.
We write $N_{\idc{u}{v}}$ for the total number of links
placed in region $\B{u}{v}$, noting that the marginal distribution of each 
$N_{\idc{u}{v}}$ is $\bino(K, \pB{u}{v})$.

\begin{lemma} \label{lem:full}
We can bound the probability that any of the regions contain a significantly different
number of links to that expected at random:
$$ \pr \left( \max_{(\idc{u}{v}) \in \SSC} \left| N_{\idc{u}{v}} - 
K \pB{u}{v} \right| \geq 
K^{\eta} \right) \leq K^{1-2\eta}.
$$
\end{lemma}
\begin{proof} A standard argument using 
the union bound and Chebyshev gives
\begin{eqnarray*}
\pr \left( \max_{(\idc{u}{v}) \in \SSC} \left| N_{\idc{u}{v}} - 
K \pB{u}{v} \right| \geq 
K^{\eta} \right) 
& = & \pr \left( \bigcup_{ (\idc{u}{v}) 
\in \SSC} \left\{ \left| N_{\idc{u}{v}} - K \pB{u}{v} \right| 
\geq K^{\eta} \right\} 
\right)  \\
& \leq & \sum_{ (\idc{u}{v})  \in \SSC} \pr \left( \left| N_{\idc{u}{v}}  
- K \pB{u}{v}  \right| \geq K^{\eta} \right) \\
& \leq & \sum_{ (\idc{u}{v})  \in \SSC} \frac{\var(N_{\idc{u}{v}})}
{K^{2 \eta}} \\
& \leq & K^{1-2 \eta},
\end{eqnarray*}
 since $\var( N_{\idc{u}{v}} ) = K \pB{u}{v}  (1- \pB{u}{v} ) 
\leq K \pB{u}{v} $,
so that  
$\sum_{ (\idc{u}{v})  \in \SSC} \var( N_{\idc{u}{v}} ) \leq K$.
\end{proof}
\subsection{Matching links} \label{sec:matching}
We now complete the proof of Theorem \ref{thm:main} -- recall that we consider 
uniform node distributions $\pr_T$ and $\pr_R$ on a bounded domain $\D$ with
smooth boundary.

\begin{proof}{\bf of Theorem \ref{thm:main}}
The
total sum capacity 
$\csum \leq I_M + J_M$, where $I_M$ is the contribution from matched pairs of links
and $J_M$ is the contribution from unmatched links.
We will consider $M$ growing as a power of $K$, but for now, it is enough to regard $M$ as
fixed.

We pair up the remaining edges in $\SSC \setminus (\SSspine \cup \SSedge)$, working
orthant by orthant.
In particular, 
the matching between $\idc{u}{v}$ and $\idc{u-e}{v+e}$
is one-to-one, and each region is counted at most once. 

For each $\vc{e} \in \{-1,1 \}^D$, we can define the function
$\prj{e}$ by $\prj{e}(\vc{w}) = \vc{w} \cdot \vc{e}$.
The key
observation is that
for each $(\idc{u}{v}) 
\notin \SSspine$, if $\vc{e} = E(\idc{u}{v})$
then the inner product 
$0 \leq \prj{e}(\vc{v} - \vc{u}) \leq \prj{e}((\vc{v} + \vc{e}) - (\vc{u} - \vc{e}))
\prj{e}(\vc{u} - \vc{v}) + 2 D$, so that $(\idc{u-e}{v+e}) \notin \SSspine$.
We can sort the regions $\B{u}{v}$ by value of $\prj{e}(\vc{u} - \vc{v})$,
at each stage adding some $(\idc{u}{v})$ with the lowest value 
of $\prj{e}(\vc{u} - \vc{v})$
that has not yet been matched to the set $\SSbody$. We depict this matching in Figure \ref{fig:spine}.
\begin{center}
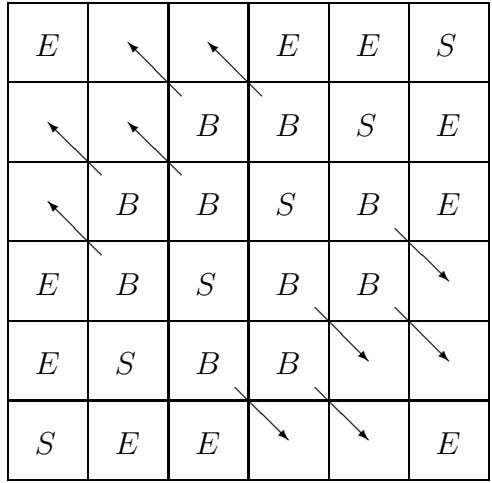
\begin{figure}[!htbp]
\centering
\begin{picture}(185,185)(-2,-2)
\multiput(0,0)(30,0){7}{\line(0,1){180}}
\multiput(0,0)(0,30){7}{\line(1,0){180}}
\multiput(10,10)(30,30){6}{$S$}
\multiput(70,40)(30,30){3}{$B$}
\multiput(100,40)(30,30){2}{$B$}
\multiput(40,70)(30,30){3}{$B$}
\multiput(40,100)(30,30){2}{$B$}
\multiput(85,35)(30,30){3}{\vector(1,-1){20}}
\multiput(115,35)(30,30){2}{\vector(1,-1){20}}
\multiput(35,85)(30,30){3}{\vector(-1,1){20}}
\multiput(35,115)(30,30){2}{\vector(-1,1){20}}
\put(40,10){$E$}
\put(70,10){$E$}
\put(160,10){$E$}
\put(10,40){$E$}
\put(10,70){$E$}
\put(10,160){$E$}
\put(160,100){$E$}
\put(160,130){$E$}
\put(100,160){$E$}
\put(130,160){$E$}
\end{picture}
\caption{Partition of regions into $\SSspine$, $\SSedge$ and $\SSbody$. 
This illustrates the case where $\D = [0,1]$ and $M=6$, and we plot the regions
as squares on $\D \times \D = [0,1]^2$.
We label
regions in $\SSspine$ by $S$, regions in $\SSedge$ by $E$ and regions in $\SSbody$
by $B$, with an arrow into the region they are matched with. \label{fig:spine}}
\end{figure}
\end{center}

This means that we can select a set of regions $\SSbody$ such that for 
$\idc{u}{v} \in \SSbody$, the $\B{u}{v}$ and $\B{u-e}{v+e}$ between them cover
$\SSC \setminus (\SSspine \cup \SSedge)$, that is
$$ \bigcup_{\idc{u}{v} \in \SSbody} \left( \B{u}{v} \cup \B{u-e}{v+e} \right)
= \SSC \setminus (\SSspine \cup \SSedge). $$

For $\D$ with volume $V$ and each $(\idc{u}{v}) \in \SSbody$  the probabilities
$\pB{u}{v} = \pB{u-e}{v+e} = 1/(V^2 M^{2D})$, since the two boxes do not intersect
the boundary of $\D \times \D$.
By assumption
there are at least $K/(V^2 M^{2D}) - K^{\eta}$ links in each of the regions
$\B{u}{v}$ and $\B{u-e}{v+e}$,
since Lemma \ref{lem:full} shows
that the probability that this does not occur is $\leq K^{1-2\eta}$.
We choose $K/(V^2 M^{2D}) - K^{\eta}$
links randomly from each pair of sets and match them, with Lemma \ref{lem:boxing} implying that  
each pair of matched links contributes
at most $\log(1 + 2 f(\euc(\idc{u}{v})))$ to the sum capacity. 
Overall, 
the total contribution to the sum capacity from all the matched links satisfies
\begin{equation} \label{eq:matchedsum}
  I_M \leq \sum_{\idc{u}{v} \in \SSbody} \frac{K}{V^2 M^{2D}} 
\log(1 + 2 f(\euc(\idc{u}{v}))).\end{equation}
By the definition of Riemann integration, by picking $M$
sufficiently large, the term $I_M/K \leq \half \ep \log(1 + 2 \SNR) + \epsilon/2$,
for $\epsilon$ arbitrarily small.

Next we control $J_M$, the contribution from the unmatched links. Specifically, without loss of generality, 
if $\D$ is bounded with volume $V$, 
we can assume $\D \times \D  \subseteq [0,L]^{2D}$ for some $L$.
\begin{enumerate}
\item
 We do not attempt
to match some regions because
they belong in $\SSedge$ or $\SSspine$. 
\begin{enumerate}
\item We assume that the boundary of $\D$ is sufficiently smooth that there exists a finite
$A$ (`surface area') such that
$|\SSedge| \leq A M^{2D-1}$  for all $M$. (For example if $\D = [0,1]^D$ then
$|\SSedge| \leq 4D M^{2D-1}$, since there are $2D$ co-ordinates that can take values $0$ or
$M-1$,
and then $M^{2D-1}$ values for the remaining co-ordinates).
\item
The number of regions in $|\SSspine| \leq D (LM)^{2D-1}$, since there
are $D$ co-ordinates which can agree, and at most $LM$ possible values the remaining
co-ordinates can take. 
\end{enumerate}
Overall, there are at most $(A + D L^{2D-1}) M^{2D-1}$ regions we do not attempt to match.
Each region we do not attempt to match contains
at most $K/(V^2 M^{2D}) + K^{\eta}$ links.
\item
We attempt to perform matching between at most
$(LM)^{2D}$ regions, with at most $2 K^{\eta}$ unmatched links
remaining from each.
\end{enumerate}
In total we deduce 
there are $(A + D L^{2D-1})(K/(V^2 M) + K^{\eta} M^{2D-1})
+ 2 (LM)^{2D} K^{\eta}$ unmatched links. Writing $\beta =
1/(3(2D+1))$, and
choosing $M = K^{3 \beta(1-\eta)}$, and for example taking $\eta = 2/3$, there are 
$O \left(K^ {1 - \beta} \right)$ unmatched links.

The single user capacity bound (see for example  \cite[Equation (6.4)]{tse})
tells us that reliable rates satisfy:
\begin{equation} \label{eq:single}
\R{i} \leq \log( 1 + \SNR_i) \leq \log(1 + 2 \SNR_i) = 2 S_{ii} \leq 2 \max_i
S_{ii},\end{equation}
so that $J_M = c_1 K^{1-\beta} \max_{1 \leq i \leq K} S_{ii}$ for some $c_1$.
Hence overall, the probability
\begin{eqnarray*}
\pr( J_M/K \geq \epsilon/2)
& \leq & \pr \left(  \max_{1 \leq i \leq K} S_{ii} \geq  \frac{\epsilon K^{\beta}}{2 c_1} \right),
\end{eqnarray*}
which tends
to zero by  Equation (\ref{eq:extreme}).
\end{proof}

Note that the upper bound in Equation (\ref{eq:ozgur}) obtained by {\"O}zg{\"u}r, L{\'e}v{\^e}que, and 
Tse \cite{ozgur} essentially has the extra factor of $\log K$ since the bound is only made up
on the term $J_M$. It is precisely the matching argument that gives rise to the $I_M$ term
which has reduced the order of the bound, as we take advantage of the extra randomness provided
by placing transmitter and receiver nodes separately.

Note that Equations (\ref{eq:extreme}) and (\ref{eq:single}) together give probabilistic
bounds on $\max_{1 \leq i \leq K} \R{i}$. Specifically, since they prove that
$\max_{1 \leq i \leq K} \R{i} = O_{\pr}(\log K)$, they control the extent to which a
large sum capacity can be achieved by a small number of links that operate at particularly
high capacity. This suggests that in this case the sum capacity is not too unfair a
measure of network performance.
\section{Future work and extensions} \label{sec:future}
We briefly comment on some extensions of Theorem \ref{thm:main} to more
general models.

\subsection{Non-uniform node distributions} \label{sec:non-unif}

We would like to consider more general distributions $\pr_R$ and $\pr_T$, rather than
simply assuming that these distributions are uniform. Indeed, we might wish to extend
to a situation where $\pr_{R,T}$ have a joint distribution from which we sample independently
 to find receiver and transmitter distributions.

The main issue that arises is whether we can quantize the joint distribution into regions
$\B{u}{v}$ which can be uniquely paired off with other regions $\B{u'}{v'}$, with the paired
regions having equal probability, and with bounds on the sum of reliable rates being
possible according to results such as Lemma \ref{lem:boxing}.

One case that certainly works is that where the co-ordinates of $\pr_R$ and $\pr_T$ are independent. 
In this case we can extend Equation (\ref{eq:boxdef}) to obtain
\begin{eqnarray*} 
\B{u}{v} & = & \left\{ (\idc{x}{y}):   
q_{\RR,l}(u/M) \leq x^{(l)} < q_{\RR,l}((u+1)/M), \right. \\
& & \left. \;\;\;\;\;\;\;\;\;\;\;\;\; q_{\TT,l}(v/M) \leq y^{(l)} < q_{\TT,l}((v+1)/M),
 \mbox{ \; for all $l$} \right\},
\end{eqnarray*} 
where $q_{\RR,l}(x)$ is the $x$th quantile of the distribution of the $l$th component of the
receivers, and
$q_{\TT,l}(x)$ is the $x$th quantile of the distribution of the $l$th component of the
transmitters.

The proof of Theorem \ref{thm:main} follows exactly as before in this case.

\subsection{Random fading model} \label{sec:randfad}
We briefly remark on an adaption of Definition \ref{def:transprot} that would
have the same independence structure, while introducing random fading into the
model. That is, we could set
   $$   \INR_{ij} = M_{ij} f(\euc(t_i,r_j)), $$
where the $M_{ij}$ are IID random variables with a density. Under our node placement
model, the $\SNR_i$ will again be IID, making this model  tractable in much the same
way.
In particular we can extend the tail behaviour bounds given in Lemma \ref{lem:tailbd}
to this case, adjusting the constant to take account of the random fading term.
\begin{lemma} \label{lem:tailbd2}
Consider an  IID network with spatially separated $\pr_T$ and 
$\pr_R$, with signals decaying at rate $\alpha$. If the fading random variables 
$M_{ij}$ have finite mean then
the probability that the maximum of the $K$ 
random variables $S_{ii}$ is large tends to zero:
$$
 \lim_{K \tends} \pr \left( \max_{1 \leq i \leq K} S_{ii} \geq \max
\left( \frac{2 \alpha}{\dsep},1 \right) \log K
\right) = 0. $$
\end{lemma}
\begin{proof}
An equivalent of Equation (\ref{eq:exptail}) holds, since as before, for any $u \geq 1$,
\begin{eqnarray*}
\pr( S_{ii} \geq u) 
& \leq & \pr( M_{ii} f(\euc(T_i, R_i)) \geq \exp(2u)/3) \\
& \leq & \pr \left( \{ M_{ii} \geq \exp(u)/\sqrt{3} \} \cup
\{ f(\euc(T_i, R_i)) \geq \exp(u)/\sqrt{3} \} \right) \\
& \leq & \pr( M_{ii} \geq \exp(u)/\sqrt{3} )  +
\pr( \cdec/\euc(T_i, R_i)^{\alpha} \geq \exp(u)/\sqrt{3}) \\
& = & \sqrt{3} \ep M_{ii} \exp(-u) +  
\pr( \euc(T_i, R_i) \leq (\cdec \sqrt{3})^{1/\alpha} \exp(-u/\alpha)) \\
& \leq & \sqrt{3} \ep M_{ii} \exp(-u) + \csep (\cdec \sqrt{3})^{\dsep/\alpha} \exp(-u 
\dsep/\alpha)).
\end{eqnarray*}
so the result follows exactly as before, using the union bound.
\end{proof}
The key to proving convergence in probability of $\csum/K$ is to show that
matching is possible between elements of $\B{u}{v}$ and $\B{u-e}{v+e}$, as
before. In the case of deterministic fading, any links $(R_i,T_i) \in \B{u}{v}$
and $(R_j,T_j) \in \B{u-e}{v+e}$ could be matched, since the proof of 
Lemma \ref{lem:boxing} showed that in this case
\begin{equation} \label{eq:orderdist}
\euc(T_j, R_j) \geq \euc(T_j, R_i) \geq \euc( T_i, R_i).
\end{equation}
In the case of random fading, this is not enough to control the relevant values of
$\INR$. However, Lemma \ref{lem:match2} shows that we
can match a high proportion of links, by looking for $(i,j)$
such that
\begin{equation} \label{eq:ordern}
M_{jj} \leq M_{ji} \leq M_{ii},
\end{equation}
which can be combined with Equation (\ref{eq:orderdist}) to deduce that $\SNR_j
\leq \INR_{ji} \leq \SNR_i$, so that again the sum capacity of the relevant two user
channel $\leq \log(1 + 2 \SNR_i)$. Note that it is enough for our purposes to consider the 
case of uniform $M_{ij}$ with densities, since only the ordering between random
variables matters in Equation (\ref{eq:ordern}).  We give a technical lemma that will imply the control that we require.
\begin{lemma} \label{lem:match2}
Consider a bipartite graph with $n$ vertices in each part
which we refer to as $\AAA = (A_1, \ldots, A_n)$ and $\BBB = (B_1, \ldots, B_n)$ 
respectively. All the vertices are labelled with independent $U[0,1]$ random 
variables, with
$A_i$ labelled by $U_i$ and $B_j$ labelled by $V_j$. The bipartite
graph has an edge from
$A_i$ to $B_j$ iff $U_i \leq W_{ij} \leq V_j$, where $W_{ij}$ are $U[0,1]$, 
independent of $(\vc{U},\vc{V})$ and each other.

For any $\gamma \geq 2/3$,
there exists a matching of all but $O(n^{\gamma})$  vertices, with probability
$\geq 1 - 5 \exp( -2n^{2 \gamma -1})$.
\end{lemma}
\begin{proof} 
We 
throw away the
$3 n^{\gamma}$ vertices with the biggest values of $U_i$ and the $3 n^\gamma$ vertices
with the lowest values of $V_j$.
This leaves  new sets $\overline{\AAA}$ and $\overline{\BBB}$ each
of size $N = N(n) = n - 3 n^\gamma$. We will show that there exists a matching between 
$\overline{\AAA}$ and $\overline{\BBB}$ with high probability, again using 
Equation (\ref{eq:walkup}) and controlling the probability of blocking pairs.
We condition on the event  
\begin{equation} \label{eq:dkw}
 \left\{ \sup_t \left| \frac{\# \{ i: U_i \leq t \}}{n} - t \right| \leq n^{\gamma-1} \right\}
\bigcup \left\{ \sup_t \left| \frac{\# \{ i: V_i \leq t \}}{n} - t \right| \leq n^{\gamma-1}
\right\},
\end{equation}
since Massart's form of the Dvoretzky--Kiefer--Wolfowitz theorem \cite{massart}
 tells us that this does not take place with probability $\leq 
4 \exp( - 2n^{2\gamma-1})$.

Conditional on the event (\ref{eq:dkw}), for any $k$
we know there are more than $n-k-2n^\gamma$ values of $U_i$
which are less than $1-k/n - n^{\gamma-1}$ 
. Equivalently,
there are fewer than $k + 2 n^{\gamma}$ values of $U_i$ larger than $1-k/n - n^{\gamma-1}$.
Since we throw away the largest $3 n^{\gamma}$ values of $\AAA$, 
any subset of $\AAA_S \subset \overline{\AAA}$ of size $k$ has at least
$n^{\gamma}$ vertices with $U_i$ values less than $1-k/n - n^{\gamma-1}$ (`small vertices').

By a similar argument, any subset $\BBB_S \subset \overline{\BBB}$ of size $N-k+1$ 
has at least $n^{\gamma}$ vertices with $V_j$ values greater than $1-k/n + n^{\gamma +1}$
(`large vertices'). See Figure \ref{fig:dkw} for a depiction of these events.

\begin{center}
\begin{figure}[!htbp]
\centering
\begin{picture}(400,130)
\put(8,28){$U$}
\put(8,98){$V$}
\put(20,30){\line(1,0){350}}
\put(20,100){\line(1,0){350}}
\multiput(325,20)(1,0){46}{\line(0,1){20}}
\multiput(20,90)(1,0){46}{\line(0,1){20}}
\put(200,10){\line(0,1){100}}
\put(190,0){$1-k/n$}
\put(185,10){\line(0,1){30}}
\put(75,5){$1-k/n - n^{\gamma-1}$}
\put(160,10){\vector(1,0){20}}
\put(215,90){\line(0,1){30}}
\put(245,115){$1-k/n + n^{\gamma-1}$}
\put(240,115){\vector(-1,0){20}}
\multiput(175,28)(5,0){30}{x}
\multiput(65,97)(5,0){32}{x}
\end{picture}
\caption{Position of vertices in subsets $\AAA_S$ and $\BBB_S$ forced by
throwing away $3n^{\gamma}$ 
largest values of $U$ and $3n^{\gamma}$ smallest values of $V$. \label{fig:dkw}}
\end{figure}
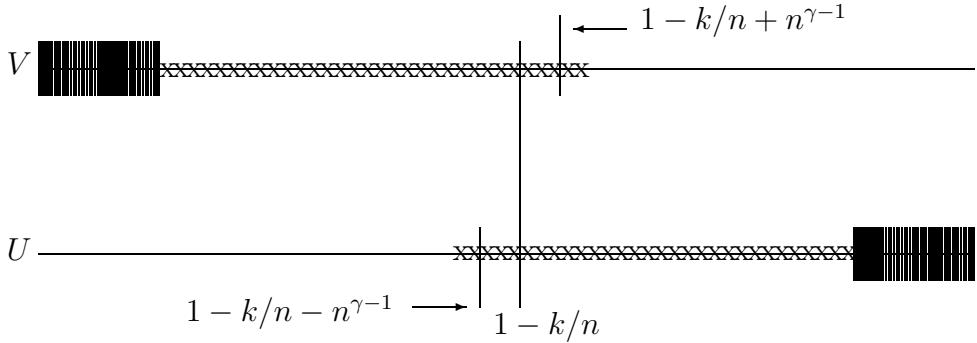
\end{center}

There is an edge between each of these small vertices in $\AAA_S$ and large vertices
in $\BBB_S$ independently with probability at least $2 n^{\gamma-1}$.
Hence, the probability that a particular $\AAA_S$ and $\BBB_S$ form a blocking pair is less
than 
$(1 - 2 n^{\gamma-1})^{ n^{2\gamma}} \leq \exp( - 2 n^{3 \gamma - 1}).$
Substituting in Equation (\ref{eq:walkup}) the probability of no matching
is 
$$ \leq \sum_{k=1}^N \binom{N}{k} \binom{N}{N-k+1} \exp( - 2 n^{3 \gamma - 1}) =
\binom{2N}{N+1} \exp(-2 n^{3 \gamma-1}) 
\leq 2^n \exp(-2 n^{3 \gamma-1}),
$$
and the result follows since $3 \gamma - 1 \geq 1$, combining with the probability of
(\ref{eq:dkw}) failing to occur.
\end{proof}
The remainder of the proof of Theorem \ref{thm:main} carries over as before. We
need only alter Section \ref{sec:matching}, and this can be done  the 
increase in numbers of unmatched vertices remains sublinear in $K$.
%
\subsection{Constructive algorithm} \label{sec:algorithm}
Although Theorem \ref{thm:main} only gives a result concerning average performance of large networks,
it does suggest some techniques that can be used to approximate the sum capacity of any
particular 
Gaussian interference network. If the network is created via a spatial model, then we can
attempt to match cross-links into $\epsilon$-bottleneck channels using the constraints on
spatial position described in the proof of Theorem \ref{thm:main}, deducing bounds as a
result.

However, even if we are only presented with the values of $\SNR_i$ and $\INR_{ij}$, it may
be possible to find bounds on sum capacity using the insights given by Lemma \ref{lem:twouser}.
Using the interference alignment scheme of \cite{nazer}, we know that a lower bound on 
$\csum$ is given by $\sum_i \half \log(1 + 2\SNR_i)$.

One possible algorithm to find an upper bound works as follows:
\begin{enumerate}
\item Sort the indices by value of $\SNR_i$, and for some $M$, partition the transmitter--
receiver links into $2M$ categories $B_{r}$ of approximately equal size $K/(2M)$.
\item For each $M$, we attempt to match links between $B_{2M-1}$ and $B_{2M}$, considering the
bipartite graph between them. 
\begin{enumerate}
\item We add an edge to the bipartite graph between $j \in B_{2M-1}$
and $i \in B_{2M}$ if $\SNR_j \leq \INR_{ji} \leq \SNR_i$.
\item We look for a maximal matching on the bipartite graph, using (for example) the
Hopcroft--Karp algorithm \cite{hopcroft}, 
which has complexity $\sqrt{V} E$, where $V$ is the number of
vertices and $E$ the number of edges.
\end{enumerate}
\item By Lemma \ref{lem:twouser}
each edge $(j,i)$ in each maximal matching contributes $\log(1 + \INR_{ji} + \SNR_i)$ as an
upper bound on the sum capacity, and each unmatched vertex $i$ simply contributes the single
user upper bound of $\log(1+\SNR_i)$
\end{enumerate}

By varying the size of $M$, this algorithm will find a range of
upper bounds, of which we can choose the tightest. We want $M$ large enough that
categories $B_r$ each contain a narrow range of $\SNR$ values, but $M$ small enough
that there are plenty of points in each range $B_r$ to ensure a large maximal matching.
\section{Conclusions}
In this paper, we have 
deduced sharp bounds for the sum capacity
of a Gaussian interference  network.
Our main contribution comes
through the upper bound, which uses arguments based on controlling 
the position of pairs of vertices to
match them
into bottleneck states. Although our main result is proved under
the assumption of deterministic fading, with signal strength decaying as a function of distance,
in Section \ref{sec:randfad} we describe an extension to a model
with random fading.

\section*{Acknowledgements}
M.~Aldridge and R.~Piechocki would like to thank Toshiba Telecommunications Research
Laboratory and its directors for supporting this work.
The authors would like to thank Justin Coon and Magnus Sandell of Toshiba for
their advice and support with this research.

\end{document}